\documentclass[sigconf,nonacm]{acmart}
\usepackage{algorithm}
\usepackage{algorithmic}

\usepackage{amsmath}

\usepackage{amssymb}
\usepackage{booktabs}
\newcommand{\rPESQ}{3.03}
\newcommand{\rDenoise}{0.98}
\newcommand{\rCodec}{0.98}
\newcommand{\rNokey}{0.51}
\newcommand{\rForge}{0.04}
\newcommand{\rHeldTPR}{0.02}

\newcommand{\name}{KeyBound}


\title{\name: Keyed and Host-Bound Learned Audio Watermarking\texorpdfstring{\\}{ }for Speech Provenance}

\author{Bangshuo Zhu}
\authornote{These authors contributed equally.}
\affiliation{%
  \institution{University of New South Wales}
  \country{Australia}
}
\author{Yuxin Cao}
\authornotemark[1]
\affiliation{%
  \institution{National University of Singapore}
  \country{Singapore}
}
\author{Weifei Jin}
\affiliation{%
  \institution{Duke University}
  \country{USA}
}
\author{Fusen Guo}
\affiliation{%
  \institution{University of New South Wales}
  \country{Australia}
}
\author{Huadong Mo}
\affiliation{%
  \institution{University of New South Wales}
  \country{Australia}
}
\author{Jingling Xue}
\affiliation{%
  \institution{University of New South Wales}
  \country{Australia}
}
\author{Wei Song}
\affiliation{%
  \institution{University of New South Wales}
  \country{Australia}
}

\begin{document}

\begin{abstract}
Audio watermarking is a proactive route to attributing synthetic speech to its source. Learned audio watermarks are typically judged by payload recovery after a fixed catalog of signal distortions such as noise, compression, filtering, and resampling. That test is necessary but not sufficient for provenance. A mark offered as evidence of origin should not be readable by an unauthorized party, should not be transferable to unrelated audio, and should not vanish when the recording is re-synthesized by a modern generative model. We present \name{}, a learned audio watermark that restores the two ingredients classical watermarking supplied and learned schemes set aside, a secret key and a host-aware carrier. \name{} masks the payload with a secret key and embeds the masked bits through a carrier modulated by a frozen spectral representation of the host, so the key governs payload access while the host-conditioned carrier resists direct transplantation. A key-independent presence head lets any party detect a mark, whereas only a key holder reads its attribution, and under the single-sample uniformity assumption a wrong-key decode clears our verification rule with probability at most $2.1\times10^{-3}$. On LibriSpeech against WavMark, AudioSeal, and Timbre, \name{} holds 1.00 detection accuracy and 0.98 bit accuracy under a spectral denoiser that costs every baseline its detection, decodes at chance without the key, and rejects transplanted carriers. Detection further transfers to held-out DAC and BigVGAN re-synthesis, though exact payload recovery degrades. Speech provenance is thus better posed as a keyed, host-bound attribution problem than as the recovery of a payload under a catalog of signal distortions fixed in advance.
\end{abstract}

\maketitle

\section{Introduction}
Advances in text-to-speech and voice conversion have made synthetic speech hard to distinguish from genuine recordings, for human listeners and automatic detectors alike \citep{valle,voicebox,aasist}. The same technology that serves accessibility and entertainment enables impersonation, financial fraud, and fabricated spoken evidence, making the origin of a recording a question that must stay answerable afterwards. Audio watermarking answers proactively \citep{audioseal,wavmark,latentwm}, embedding an imperceptible signal at creation so that origin can later be read from the waveform itself.

Most schemes share an encoder--detector design in which a generator adds a low-energy carrier to the host and a jointly trained detector recovers the payload, with recent work reporting strong recovery under noise, filtering, resampling, and compression \citep{wavmark,audioseal,silentcipher,xattnmark}. That evaluation is necessary but insufficient for provenance. A fixed distortion catalog models only accidental degradation, not a party who means to read the payload, erase the mark, or transfer attribution to audio its owner never produced. Once a watermark is offered as evidence of origin, these adversarial questions stop being optional stress tests and become the task itself.

They also expose a limitation of the design rather than of the benchmark. In most schemes no secret protects the carrier, so any party who can run the public detector reads the payload and can turn that access into an estimate of the mark and a means to suppress it. The carrier is also largely independent of the host, which makes transplantation possible, since a pattern recovered from one recording may verify once moved onto another and so fabricate an attribution, one of the forgery risks benchmarked for audio watermarks \citep{audiomarkbench}. Because the carrier is engineered to be quiet, it also tends to live in exactly the signal detail a modern re-synthesis model discards or rewrites, and a single regeneration step erases it \citep{deepshallow}, while a learned remover trained without detector access strips it outright \citep{harmonicattack}, with provable removal reported in the image domain \citep{zhao2024invisible}. These three failures share one cause. The mark is bound neither to a secret nor to the signal it marks, precisely the two ingredients classical watermarking supplied and learned schemes set aside, a secret key as in spread-spectrum watermarking \citep{cox1997secure} and a host-aware carrier as in quantization index modulation \citep{chenwornell2001} and the dirty-paper coding \citep{costa1983,gelfandpinsker1980,moulinosullivan2003} from which that whole family of codes ultimately descends.

We propose \name{}, a keyed and host-bound audio watermark that supplies both forms of binding. It retains the learned generator--detector structure of recent watermarks and changes only what the carrier depends on. The payload is first masked with a secret key, so the embedded bits carry no readable attribution without it at all. The carrier is then modulated by a frozen spectral representation of the host, making the mark signal-dependent instead of a pattern that can be lifted and reused. At verification the key recovers the payload, while a key-independent presence head decides on its own whether audio is marked at all. The only keyed learned audio watermark we know of, WAKE \citep{xu2025wake}, keys its whole transform and leaves the carrier host-independent, whereas \name{} keys the payload alone behind a public presence head and anchors the carrier in the very signal that it marks.

This design separates three properties that existing evaluations tend to conflate. No-key payload protection comes from the key mask, which randomizes the embedded bits for anyone but the verifier who can strip the mask before decoding. Transplant resistance comes from the host anchor, whose effect we measure empirically against direct transfer of a carrier onto an unrelated host. Robustness to re-synthesis is of a different kind altogether, since no key confers it and it must instead be learned under channels that rewrite low-level acoustic detail. Our experiments show that the distinction is not academic, as highly transparent additive marks survive the standard distortion catalog and still vanish under a single re-synthesis, whereas \name{}, trained explicitly for that channel, keeps its detection intact throughout the suite.

We make three major contributions:
\begin{itemize}
\item \textbf{Reframing speech provenance.} We argue that provenance demands payload keying and host binding, identify both as absent from current learned audio watermarks, and separate no-key protection, transplant resistance, and re-synthesis robustness into properties each evaluated on their own terms rather than simply assumed together.

\item \textbf{\name{} and a verification rule.} A keyed, host-bound learned watermark with a key-independent presence head, together with a provenance-oriented verification rule $V_A$ carrying a wrong-key acceptance bound (Proposition~\ref{prop:keyed_soundness}), under which a wrong-key decode clears the rule with probability at most $\beta\approx2\times10^{-3}$.

\item \textbf{Evaluation.} We test \name{} on LibriSpeech against WavMark, AudioSeal, and Timbre across signal distortion and re-synthesis, and additionally under no-key decoding and carrier transplantation. It reaches 1.00 detection accuracy and \rDenoise\ bit accuracy under a spectral denoiser that costs every baseline its detection, decodes at chance without the key, and rejects transplanted carriers outright.
\end{itemize}

\section{Related Work}

\paragraph{Learned audio watermarking.}
Most learned audio watermarks are post hoc, in that a generator adds a low-energy carrier to the host waveform and a jointly trained detector recovers the payload from it. Examples include invertible window schemes \citep{wavmark}, localization-trained generator--detector pairs \citep{audioseal}, psychoacoustic schemes \citep{silentcipher}, timbre-based schemes \citep{timbre}, dual-embedding invertible networks \citep{ideaw}, and cross-attention detectors \citep{xattnmark}. They advance capacity and imperceptibility and report robustness to a fixed distortion set, yet in nearly all of them the carrier is host-independent and unkeyed, so neither no-key payload protection nor host binding enters the design. A separate line embeds the watermark inside the generative model so that every sample it produces is marked \citep{latentwm}, which helps whoever controls that model but not the verifier who is handed only a waveform, the setting we study.

The principal exception is WAKE \citep{xu2025wake}, a keyed learned audio watermark aimed at access control and multi-watermark capacity, which gates an invertible embedding network by the key so that decoding under a wrong key falls to chance. Because the gate governs the whole transform, presence and payload are alike key-dependent. The carrier remains host-independent, the keyed claim is purely empirical, and the evaluation covers classical signal edits alone, testing neither removal nor false attribution. \name{} masks only the payload behind a key-independent presence head, so any party may detect a mark while only a key holder reads its attribution. It anchors the carrier in a frozen host representation, driving transplant verification to near zero, supplies a formal no-key statement and a wrong-key acceptance bound (Proposition~\ref{prop:keyed_soundness}), and, trained under neural re-synthesis, holds detection under a spectral denoiser where additive schemes in WAKE's own family collapse \citep{deepshallow}.

\paragraph{Classical foundations.}
Both properties that \name{} restores are old watermarking ideas. Carrier secrecy is the premise of spread-spectrum watermarking, whose secret spreading key hides the mark and protects it \citep{cox1997secure}. Host dependence is the premise of quantization index modulation, whose codes embed the payload relative to the cover signal with a minimum-distance margin \citep{chenwornell2001}. Quantization index modulation is itself an instance of communication with side information at the encoder, as in Costa's dirty-paper coding \citep{costa1983} and the Gelfand--Pinsker channel \citep{gelfandpinsker1980}, analyzed for information hiding by \citet{moulinosullivan2003}. Learned audio watermarks traded both away for transparency and distortion-robustness, whereas \name{} carries keying and host binding into a learned design and trains it under the re-synthesis attacks that the classical codes of that era never had to face.

\section{Problem Formulation}
We watermark a host waveform $x\in\mathbb{R}^{T}$ with an $n$-bit payload $m\in\{0,1\}^{n}$ under a secret key $\kappa\in\{0,1\}^{L}$ held only by the embedder and an authorized verifier. An embedder $E$ produces a marked signal $x_w=E(x,m,\kappa)$ within a perceptual budget, and an adversary may apply any transformation $a$ from an admissible class within a distortion budget. A key-independent presence head returns a score $s(y)\in[0,1]$, whereas an authorized decoder $D$, available to a key holder alone, returns a payload estimate $\hat m(a(x_w),\kappa)$. Beyond payload recovery under benign $a$, we ask three things of $(E,D)$. \emph{Robustness} operates at two levels, since presence robustness asks only that the marked/unmarked decision survive an admissible transformation, whereas payload robustness further asks $\hat m=m$. Detection accuracy (DA) measures the former and is our primary metric, bit accuracy (BA) the latter. \emph{No-key payload protection} requires that without $\kappa$ the detector output decode at chance (Proposition~\ref{prop:single_sample}), a statement about detector-side recovery, not about multi-message secrecy of the waveform. \emph{Resistance to false attribution} asks that a party without $\kappa$ should not make audio the verifier accepts under a chosen payload, which we test in the case where a carrier is lifted from a marked recording onto unrelated audio. The key governs payload access alone and confers no robustness against removal, which we obtain instead from the host-bound carrier and from training under attack channels.

Keyless post-hoc schemes leave two gaps for provenance. Absent a secret key, a public detector reads $m$ straight from $x_w$, and no amount of distortion training alters that. Absent host binding, a carrier lifted from one recording may verify on another whenever it is reusable, as we measure directly on a leading scheme, though not every keyless scheme proves transplantable. A secret key and a host-conditioned carrier address both gaps, and none of the baselines that we evaluate supplies the two together, which is precisely the gap that \name{} is built to close.

\subsection{Threat Model}
A provider embeds a watermark into the speech it generates under a secret key $\kappa$, and a verifier holding $\kappa$ later tests a recording to establish its origin. We assume the adversary knows everything but the deployment key, so the architecture and weights of $E$, $D$, and the host analysis $C$ may be public. It may run them on arbitrary inputs and trial keys and collect any number of marked recordings, yet it can neither obtain $\kappa$ nor query a verifier that holds it. Our keyed-access claim is correspondingly single-sample, because the fixed reused mask affords no multi-message secrecy and no resistance to known-payload inference, both of which fall outside the scope of what this paper claims and of what it evaluates.

The adversary holds a marked recording and pursues one of three goals. It may \emph{read the payload}, recovering $m$ without $\kappa$. It may \emph{remove the mark}, perturbing the recording within a perceptual budget until the detector stops firing, whether by signal processing, by re-synthesis through a codec, a vocoder, or a denoiser, or by a removal network trained against the scheme. It may \emph{produce a false attribution}, presenting a recording the verifier accepts under a chosen payload, either by embedding one or by transplanting a carrier from a marked recording onto audio the owner never produced. The adversary lacks only the key, and we test each of these three goals under the attacks that we implement. An attack counts only for as long as the audio stays perceptually useful, since a recording degraded past that point has already defeated the adversary's own purpose in attacking it in the first place.

\section{Method}
\name{} keeps the generator--detector structure of post-hoc watermarks and their evaluation conventions, changing two things only, what the carrier depends on and how the payload is protected. It operates on 16~kHz audio in one-second units with a 16-bit payload, and Figure~\ref{fig:overview} sketches its embedder and detector. The secret key is held by the embedder and the authorized payload decoder, whereas the presence head runs without it and so stays available to any verifying party at all.

\begin{figure*}[t]
\centering
\includegraphics[width=\textwidth]{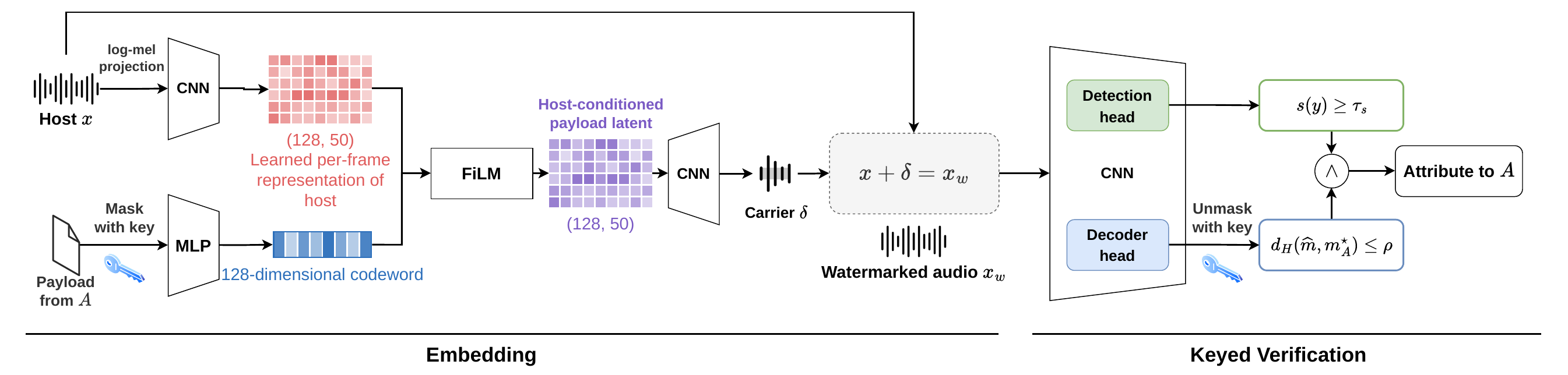}
\caption{Overview of \name{}. The payload is masked with a secret key and embedded through a carrier modulated by spectral host features $C(x)$, so the mark is both keyed and host-bound. At verification the detector produces a key-independent presence score $s(y)$, letting any party detect that audio is marked, while only a holder of $\kappa$ can unmask it. Attribution requires both, since the rule $V_A$ fires when the presence score clears $\tau_s$ \emph{and} the decode lands within $\rho$ bits of the registered payload.}
\label{fig:overview}
\end{figure*}
The key masks the payload so that the masked bits appear random to anyone who does not hold the key,
\begin{equation}
\tilde m = m \oplus \mu(\kappa), \qquad \mu(\kappa)=\kappa_{1:n},
\label{eq:xor}
\end{equation}
where $\oplus$ is the bitwise XOR and $\mu(\kappa)$ is a per-bit mask read from the key. The host analysis $C$ maps the host to latent-width features through a fixed, phase-free log-mel transform followed by a small learned projection. The generator forms the carrier by projecting the masked payload and modulating it with $C(x)$ through feature-wise linear modulation (FiLM) \citep{film},
\begin{equation}
z = \mathrm{FiLM}\big(P(\tilde m),\, C(x)\big), \qquad
\delta = \alpha\, G_\theta(z),
\label{eq:gen}
\end{equation}
where $P$ is a learned per-bit projection, $G_\theta$ is a convolutional synthesis network, and $\alpha$ scales the carrier to the perceptual budget. The marked signal adds the carrier to the host,
\begin{equation}
x_w = x + \delta.
\label{eq:marked}
\end{equation}
Because $\delta$ passes through $C(x)$ in \eqref{eq:gen}, the generator is free to make the carrier a function of the recording it marks, and such a carrier is correspondingly harder to reuse, while transplant-negative training teaches the detector to reject one moved onto a mismatched host. The masked payload remains the larger driver of the carrier and stays recoverable from it, yet the host modulation is substantial by design rather than a small perturbation, since a carrier the host barely influenced would be precisely the transferable pattern that the anchor exists to prevent, and this is exactly what the ablation that removes the anchor confirms.

\subsection{Keyed Decoding and Detection}
The detector computes a learned analysis of the received audio $y$, produces per-frame payload logits and a presence score, and pools the payload over frames weighted by presence,
\begin{equation}
h = E_\psi(y),\quad
\bar\ell = \sum_{t} \mathrm{softmax}_t \big(w_d^\top h\big)_t \,\big(W_m h_t\big),
\label{eq:pool}
\end{equation}
where $W_m$ and $w_d$ are the payload and detection heads. It then removes the key mask by flipping the sign of every logit whose bit the mask inverted, an operation that recovers the payload for a key holder and for nobody else,
\begin{equation}
\ell = \bar\ell \odot \big(1 - 2\,\mu(\kappa)\big), \qquad
\hat m = \mathbf{1}[\ell > 0].
\label{eq:decode}
\end{equation}
Detection uses the pooled presence score,
\begin{equation}
s(y)=\tfrac1{T'}\sum_t \sigma(w_d^\top h_t),
\end{equation}
calibrated once on clean unmarked audio and then held fixed.

\subsection{Training Objective}
We train the generator $G_\theta$, the detector $E_\psi$, and the projections $\phi$ end to end with the log-mel front end of $C$ frozen. Each step draws an attack channel $A$ and applies it to the marked signal before decoding, mixing the distortion catalog with re-synthesis (codec round trip, spectrogram vocoder, spectral denoiser) and with an in-loop removal network $R$ that is updated alongside the model. The full training objective then takes the form
\begin{equation}
\begin{split}
\min_{\theta,\psi,\phi} \mathbb{E}\big[ &\lambda_m \mathcal{L}_{\text{msg}} + \lambda_d \mathcal{L}_{\text{det}} + \lambda_f \mathcal{L}_{\text{forge}} \\
&+ \lambda_p \mathcal{L}_{\text{perc}} + \lambda_r \mathcal{L}_{\text{snr}} \big],
\end{split}
\label{eq:obj}
\end{equation}
where $\mathcal{L}_{\text{msg}}$ is the binary cross-entropy between the keyed decode $\ell(A(x_w),\kappa)$ and the payload $m$. The presence-head cross-entropy $\mathcal{L}_{\text{det}}$ takes marked audio as positive against three negative sets, namely clean, attacked, and remover output, which denies the head any chance of learning to equate distortion itself with the presence of a mark. The false-attribution term $\mathcal{L}_{\text{forge}}$ drives the presence score to zero for a carrier transplanted onto a shuffled host, and the final two terms bound perceptibility,
\begin{equation}
\begin{split}
\mathcal{L}_{\text{perc}} &= \mathcal{L}_{\text{mrstft}} + \mathcal{L}_{1} + 2\mathcal{L}_{\text{mel}} + \lambda_{\text{TF}}\mathcal{L}_{\text{TF}}, \\
\mathcal{L}_{\text{snr}} &= \max(0,\, \tau_{\text{snr}}-\mathrm{SNR}),
\end{split}
\label{eq:perc}
\end{equation}
where $\mathcal{L}_{\text{mrstft}}$ is the multi-resolution spectrogram loss \citep{parallelwavegan}, the psychoacoustic time--frequency masking loss $\mathcal{L}_{\text{TF}}$ places carrier energy where the host masks it, and $\tau_{\text{snr}}$ is ramped from 6 to 22~dB. The removal attack we report at test time is not $R$ but a separately trained reconstruction network with access to neither the detector nor the key. Training against the re-synthesis channels buys empirical robustness and no formal guarantee, and robustness to the fixed denoiser, where no remover appears, follows from that training and from the host-bound carrier rather than from the key, which governs payload access alone.

\subsection{Keyed Provenance Verification}
A provenance claim is a decision about a party, so we cast it as a per-party verification rule and separate what the design proves from what it merely demonstrates. Party $A$ holds a secret key $\kappa_A\in\{0,1\}^L$ known only to $A$ and an authorized verifier, registers a payload $m_A^\star\in\{0,1\}^n$, and produces marked audio $x_w^A=E(x,m_A^\star,\kappa_A)$. On a received waveform $y$ we distinguish two decisions. \textbf{Presence verification} is key-free and asks only whether the received waveform $y$ carries a mark of any kind at all,
\begin{equation} 
V_{\mathrm{pres}}(y)=\mathbf{1}[s(y)\ge\tau_s],
\label{eq:presence_verification}
\end{equation}
while \textbf{attribution verification} is keyed and asks whether $y$ carries $A$'s payload under $A$'s key,
\begin{equation}
V_A(y)=\mathbf{1}[s(y)\ge\tau_s]\wedge
\mathbf{1}[d_H(\widehat{m}(y,\kappa_A),m_A^\star)\le\rho],
\label{eq:attribution_verification}
\end{equation}
where $s(y)\in[0,1]$ is the key-independent presence score, $\hat m(\cdot,\kappa_A)$ the keyed decoder, $d_H$ Hamming distance, $\tau_s$ a presence threshold, and $\rho$ the tolerated bit-error budget, set to $\rho=2$ so that at least $n-\rho=14$ of 16 bits must agree. Presence verification is the object that detection accuracy measures, whereas attribution verification is the operational provenance decision, and it adds exactly the keyed payload check that presence alone never makes. We first establish the keyed-access guarantees that the mask affords, and then state the two further properties that we verify only empirically.

\begin{proposition}[Single-sample mask uncertainty]
\label{prop:single_sample}
For a single sample, if the mask $\mu(\kappa)$ is uniform and unknown to the adversary, then $\tilde m=m\oplus\mu(\kappa)$ is uniform in that view, and a decoder without $\kappa$ or a wrong key recovers $m$ at chance.
\begin{proof}
    The map $\mu\mapsto m\oplus\mu$ is a bijection, so an unknown uniform mask renders $\tilde m$ uniform. Recovering $m$ requires $\mu(\kappa)$, which only the verifier holds.
\end{proof}
\end{proposition}
\begin{corollary}[Fixed-target wrong-key agreement]
\label{cor:chosen_payload_attribution}
Under the single-sample uniformity assumption of Proposition~\ref{prop:single_sample}, a wrong-key decode matches any fixed target with the same probability as it matches a uniformly drawn target.
\end{corollary}

\begin{proof}
Fix any target $m_A^\star$. By Proposition~\ref{prop:single_sample} the recovered bits are uniform on $\{0,1\}^n$ in the adversary's view, and a uniform distribution assigns the same probability to agreement with $m_A^\star$ as to agreement with a uniformly drawn target.
\end{proof}

\begin{proposition}[Wrong-key acceptance bound for $V_A$]
\label{prop:keyed_soundness}
If $\hat m(y,\kappa_A)$ is uniform in the adversary's view, as Proposition~\ref{prop:single_sample} gives for a single sample under an unknown uniform mask, then
\begin{equation}
\Pr \left[V_A(y)=1\right]
\le
\beta
:=
\sum_{i=0}^{\rho}
\binom{n}{i}2^{-n}.
\label{eq:keyed_soundness}
\end{equation}
For $n=16,\rho=2$, $\beta=137/2^{16}\approx2.1\times10^{-3}$.
\begin{proof}
    By Proposition~\ref{prop:single_sample} and Corollary~\ref{cor:chosen_payload_attribution}, $\hat m(y,\kappa_A)$ is uniform on $\{0,1\}^n$ in the adversary's view, so $\Pr[d_H(\hat m,m_A^\star)\le\rho]=\sum_{i\le\rho}\binom{n}{i}2^{-n}=\beta$, and the presence indicator, being an indicator, can only lower that probability.
\end{proof}
\end{proposition}

The bound rests on the keyed payload check rather than on $\tau_s$, so inflating the presence-only false-positive rate gains nothing while 14 of $A$'s 16 keyed bits must still be reproduced. The two properties stated next are empirical.

\paragraph{Robust completeness.} $A$'s own audio should verify under every admissible $a\in\mathcal A$, that is $\Pr[V_A(a(x_w^A))=1]\ge 1-\varepsilon_a$. Because $V_A$ demands presence ($s\ge\tau_s$) and payload ($d_H\le\rho$) together, its true-accept rate tracks whichever conjunct is weaker, and the two diverge whenever re-synthesis preserves presence while destroying exact payload. Table~\ref{tab:va} reports the rate per attack, where DA stands for the presence conjunct and the bit budget stands for the payload conjunct of the same rule.

\paragraph{Transplant resistance.} $A$'s genuine carrier should not transfer, so for $\delta_A=x_w^A-x$ lifted onto an unrelated host $x'$ we want $\Pr[V_A(x'+\delta_A)=1]\approx 0$. Proposition~\ref{prop:keyed_soundness} does not imply this, since $\delta_A$ is here a legitimate carrier rather than a guess, and the property is instead an empirical consequence of the host-conditioned carrier and of the $\mathcal L_{\text{forge}}$ training term, and we measure it only against direct residual transplantation of the carrier onto an unrelated host, which is the attack that Table~\ref{tab:security} reports in full.

\paragraph{Scope.} The tolerance $\rho$ trades true acceptance against wrong-key acceptance. Proposition~\ref{prop:keyed_soundness} is single-sample and assumes an honest verifier, so it delivers keyed access rather than multi-message secrecy and bounds no adversary who adaptively queries the verifier to infer the mask. We do not claim cryptographic unforgeability, since an adversary that observes a valid masked codeword and can invoke or approximate the embedder may attempt replay or re-embedding on a new host. Preventing that needs the payload authenticated jointly with a host fingerprint, by a MAC or a signature, which we leave to future work.

\section{Experiments}
We evaluate \name{} along four axes, robustness to attack, imperceptibility, no-key decoding, and transplant-based false attribution, taking the last two as what distinguishes it.

\subsection{Experiment Setup}

\paragraph{Datasets.} We train on the 28,539 one-second clips at 16~kHz of the LibriSpeech \citep{librispeech} train-clean-100 split and test on 300 held-out clips from test-clean. Clip selection and payload draws are seeded, so all methods see the same 300 clips in the same order under the same payloads. Per-clip secret keys are drawn from the global random state rather than a fixed seed, which is immaterial because keys and payloads are uniform, and the model is trained over random keys. LibriSpeech is the corpus on which the released baselines were themselves developed and evaluated, so no method gains an advantage from a domain shift, and its speaker-disjoint split lets us hold out speakers as well as utterances. It is read English speech throughout, which bounds how far the conclusions drawn below can reasonably be carried beyond that setting.

\paragraph{Baselines.} We compare \name{} against three representative released schemes, each run from the authors' own code and checkpoints, namely the invertible window scheme WavMark \citep{wavmark}, the localization-trained AudioSeal \citep{audioseal}, and the timbre watermark \citep{timbre} built to survive voice cloning. At submission time the official XAttnMark project page\footnote{\url{https://liuyixin-louis.github.io/xattnmark/}} offered the paper, results, and demonstrations but no public inference code or checkpoint, leaving us unable to run it under our attack protocol. Timbre runs at its native 22.05~kHz with a 10-bit payload and is evaluated under the same attack algorithms at that rate, while the remaining methods and \name{} run at 16~kHz. WavMark is run on 1.1~second clips because its decoder partitions the signal into chunks of $16{,}000+\lfloor 16{,}000\times0.1\rfloor=17{,}600$ samples and needs one complete chunk, which a 1.0~second clip cannot supply, whereas every other method evaluated here runs on clips of exactly 1.0~second in length.

\paragraph{Attacks.} We group the transformations into two families. \textit{Signal distortions} comprise additive noise, low- and high-pass filtering, resampling, gain, 8-bit quantization, and MP3 at 32~kbps, the catalog the additive schemes are tuned to survive. \textit{Re-synthesis} regenerates the waveform through an EnCodec round trip at 6~kbps \citep{encodec}, a Griffin-Lim spectrogram vocoder \citep{griffinlim}, and a spectral denoiser built on noise estimation and spectral gating.

\paragraph{Metrics.} Detection accuracy (DA) \citep{demark} is the true-positive rate at a threshold fixed at the 1\% false-positive point of the presence-score ROC on unmarked clean clips, and we reserve TPR@1\%FPR for the attack-matched evaluation of Table~\ref{tab:fpr}. Bit accuracy (BA) is the fraction of payload bits recovered correctly. Imperceptibility is read from wideband PESQ \citep{pesq}, short-time intelligibility STOI \citep{stoi}, and signal-to-noise ratio (SNR), for all of which higher is better. No-key decoding is read from the bit accuracy an adversary obtains by running the detector without the key, where 0.5 is chance, and transplant-based false attribution from the rate at which a carrier lifted onto unrelated audio still verifies against the payload registered for the original recording.

\paragraph{Training protocol.} The generator and detector are SEANet-style convolutional networks \citep{seanet} with a 128-dimensional latent, the host analysis is a fixed log-mel front end with a small learned projection, and the payload is $n{=}16$ bits masked by a 16-bit key. We train with AdamW at a learning rate of $2\times10^{-4}$ for 30,000 steps on batches of 12 one-second clips. A curriculum forms the code on a loud mark before tightening the SNR target from 6 to 22~dB over the first 7,000 steps, enables the learned remover after a 7,000-step warmup, ramps the false-attribution term over 3,000 steps and the perceptual terms over 5,000 after a 2,000-step delay, and raises the attack-channel probability from 0.25 to 1. The objective weights are $\lambda_m{=}12$, $\lambda_d{=}1.5$, $\lambda_p{=}0.5$, and $\lambda_f{=}1$, with the time--frequency masking and SNR-hinge terms at unit weight, and the remover is updated once per generator step at the same learning rate under a carrier scale of $\alpha{=}1$.

\begin{table*}[t]
\caption{Robustness across the full attack suite. Bold marks the best value in each row, with DA presence and BA payload.}
\label{tab:robust}
\centering
\resizebox{0.8\textwidth}{!}{%
\begin{tabular}{lcccccccc}
\toprule
Attack & \multicolumn{2}{c}{WavMark} & \multicolumn{2}{c}{AudioSeal} & \multicolumn{2}{c}{Timbre} & \multicolumn{2}{c}{KeyBound} \\
\cmidrule(lr){2-3}\cmidrule(lr){4-5}\cmidrule(lr){6-7}\cmidrule(lr){8-9}
 & DA & BA & DA & BA & DA & BA & DA & BA \\
\midrule
Clean & 1.000 & 1.000 & 0.990 & 0.912 & 1.000 & \textbf{1.000} & \textbf{1.000} & 0.999 \\
\multicolumn{9}{@{}l}{\emph{Signal distortion}}\\
Noise (20\,dB) & 0.007 & 0.501 & 0.490 & 0.639 & 0.363 & 0.909 & \textbf{1.000} & \textbf{0.999} \\
Noise (10\,dB) & 0.000 & 0.498 & 0.287 & 0.545 & 0.003 & 0.667 & \textbf{1.000} & \textbf{0.974} \\
Lowpass (4\,kHz) & 1.000 & \textbf{1.000} & 0.993 & 0.912 & 1.000 & 0.991 & \textbf{1.000} & 0.999 \\
Highpass (500\,Hz) & 1.000 & 1.000 & 0.993 & 0.918 & 1.000 & \textbf{1.000} & \textbf{1.000} & 0.999 \\
Resample (0.8) & 0.947 & 0.970 & 0.967 & 0.898 & 1.000 & \textbf{1.000} & \textbf{1.000} & 0.999 \\
Gain ($\times0.5$) & 0.993 & 0.996 & 0.947 & 0.908 & 1.000 & \textbf{1.000} & \textbf{1.000} & 0.999 \\
Quantize (8-bit) & 0.747 & 0.864 & 0.937 & 0.852 & 0.977 & 0.999 & \textbf{1.000} & \textbf{0.999} \\
MP3 (32\,kbps) & 0.723 & 0.849 & 0.907 & 0.861 & 1.000 & \textbf{1.000} & \textbf{1.000} & 0.999 \\
Smoothing & 0.940 & 0.969 & 0.970 & 0.907 & 1.000 & 0.988 & \textbf{1.000} & \textbf{0.999} \\
Echo & 0.677 & 0.821 & 0.977 & 0.870 & 1.000 & \textbf{1.000} & \textbf{1.000} & 0.994 \\
\multicolumn{9}{@{}l}{\emph{Re-synthesis}}\\
EnCodec (6\,kbps) & 0.000 & 0.498 & 0.850 & 0.568 & 0.013 & 0.562 & \textbf{1.000} & \textbf{0.982} \\
Griffin-Lim & 0.983 & 0.988 & 0.080 & 0.524 & 1.000 & \textbf{1.000} & \textbf{1.000} & 0.702 \\
Denoiser & 0.000 & 0.498 & 0.173 & 0.545 & 0.007 & 0.712 & \textbf{1.000} & \textbf{0.983} \\
\bottomrule
\end{tabular}
}
\end{table*}

\paragraph{Computing platform.} All training and evaluation run on a single NVIDIA RTX~6000 Ada GPU (48\,GB) hosted by an Intel Xeon w7-3565X with 502\,GB of memory, under Ubuntu 22.04 LTS with CUDA~12.8, Python~3.10, PyTorch~2.8.0, and torchaudio~2.8.0. A full 30,000-step training run takes about 1.6~hours on this hardware, and a full evaluation pass over the whole attack suite takes under fifteen minutes for each of the methods.

\subsection{Main Results}

Under the spectral denoiser \name{} holds DA~=~1.00 and BA~=~\rDenoise\ where WavMark, AudioSeal, and Timbre lose detection, its payload decodes at chance without the key, and a carrier transplanted onto unrelated audio rarely verifies.

\paragraph{Robustness.} \name{} holds DA~=~1.00 on every perturbation, and under the EnCodec round trip it recovers the payload itself (BA~=~0.98) rather than presence alone, departing from earlier additive designs that lose the message to codec re-synthesis. Only Griffin-Lim reduces exact recovery (BA~=~0.70), where detection still stays at 1.00, so the mark is confirmed present on audio whose payload can no longer be attributed. Under the spectral denoiser \name{} keeps DA~=~1.00 and BA~=~0.98 where every baseline loses detection, WavMark falling to 0.00, AudioSeal to 0.17, and Timbre to 0.01. The baselines fail unevenly across codec and vocoder, each surviving one re-synthesis and collapsing on the other, whereas \name{} survives both. Because DA fixes its threshold on clean unmarked audio, we check that the attacks do not inflate it. Seven of the fourteen realize a false-positive rate at or below 0.013 and the rest reach 0.03 to 0.13, with 10~dB noise at 0.80, yet recalibrating on same-attack negatives still leaves DA at 1.000 throughout (Table~\ref{tab:fpr}), so the elevated rates reflect a shift in the absolute score level rather than any loss of separability between the two classes.

\paragraph{Imperceptibility.} The host-bound carrier and attack-channel training cost some transparency, placing \name{} at a wideband PESQ of 3.03, STOI of 0.981, and SNR of 24.6~dB. That is below WavMark (4.18) and Timbre (3.55) on PESQ, yet it is the highest intelligibility among the learned schemes, well clear of Timbre's STOI of 0.715. \name{} therefore buys detection that survives re-synthesis, and it pays for that with a price in perceptual quality that is real, measured, and openly stated.

\paragraph{No-key decoding and false attribution.} Without the key, \name{}'s detector recovers the payload at \rNokey\ bit accuracy, chance exactly as Proposition~\ref{prop:single_sample} requires, whereas all three baselines are keyless and yield their payload to anyone running the public decoder. For false attribution we lift a recovered carrier onto unrelated audio and ask whether it verifies, requiring the presence score to clear that method's own 1\%-FPR threshold and the payload to decode. \name{} verifies at \rForge\ against 0.59 for AudioSeal and 0.227 for Timbre, though WavMark also resists (0.000), so host independence alone does not settle transplantability. \name{} is nonetheless the only scheme evaluated supplying both properties at once, since each baseline either exposes its payload to any reader or accepts a transplanted carrier (Tables~\ref{tab:security} and~\ref{tab:nokey}).

\begin{table}[t]
\caption{True-accept rate of the attribution rule $V_A$ across the full attack suite, which requires both conjuncts to hold.}
\label{tab:va}
\centering
\footnotesize
\resizebox{0.9\columnwidth}{!}{%
\begin{tabular}{lcccc}
\toprule
Attack & WavMark & AudioSeal & Timbre & KeyBound \\
\midrule
Clean & 1.000 & 0.813 & 1.000 & \textbf{1.000} \\
\multicolumn{5}{@{}l}{\emph{Signal distortion}}\\
Noise (20\,dB) & 0.000 & 0.017 & 0.370 & \textbf{0.997} \\
Noise (10\,dB) & 0.000 & 0.000 & 0.007 & \textbf{0.980} \\
Lowpass & 1.000 & 0.810 & 1.000 & \textbf{1.000} \\
Highpass & 1.000 & 0.807 & 1.000 & \textbf{1.000} \\
Resample & 0.947 & 0.727 & 1.000 & \textbf{1.000} \\
Gain & 0.993 & 0.760 & 1.000 & \textbf{1.000} \\
Quantize & 0.743 & 0.533 & 0.977 & \textbf{1.000} \\
MP3 & 0.710 & 0.593 & 1.000 & \textbf{1.000} \\
Smoothing & 0.940 & 0.767 & 1.000 & \textbf{1.000} \\
Echo & 0.653 & 0.637 & \textbf{1.000} & 0.993 \\
\multicolumn{5}{@{}l}{\emph{Re-synthesis}}\\
EnCodec & 0.000 & 0.010 & 0.003 & \textbf{0.990} \\
Griffin-Lim & 0.987 & 0.000 & \textbf{1.000} & 0.200 \\
Denoiser & 0.000 & 0.000 & 0.010 & \textbf{0.973} \\
\bottomrule
\end{tabular}
}
\vspace{-2mm}
\end{table}

\paragraph{Attribution under the verification rule.} Table~\ref{tab:va} reports $V_A$ itself rather than its conjuncts. \name{} attributes at least 0.97 of clips on thirteen of the fourteen attacks, Griffin-Lim being the sole exception. The baselines attribute reliably under filtering and the milder edits, so the suite is not uniformly discriminative, yet they collapse under noise and re-synthesis, detecting or decoding but seldom both. The best any reaches there is 0.370 (Timbre at 20~dB noise), and none exceeds 0.010 under the codec or denoiser. AudioSeal never exceeds 0.813 even on clean audio, where a bit accuracy of 0.91 leaves one clip in five short of the budget. For \name{} the presence conjunct holds at 1.000 throughout, so attribution is bound by the payload conjunct alone, and under Griffin-Lim a mean BA of 0.70 leaves only 0.20 of clips clearing 14 of 16 bits, where WavMark and Timbre pass nearly intact. We accordingly claim provenance \emph{attribution} for \name{} under signal distortion, the codec, and the denoiser, and \emph{presence detection} alone under vocoder re-synthesis, where the exact payload no longer survives.

\subsection{Ablation Studies}
We remove each mechanism in isolation, all else fixed (Table~\ref{tab:ablation}). Every mechanism supplies one property, and the two that guard against an adversary cost a little imperceptibility.

\paragraph{The key mask.} The key mask carries no ablation row, since XOR masking is dimension-preserving and leaves every other metric untouched, so the no-key decode already reported captures its effect in full (\rNokey\ bit accuracy against 1.00 with the key). It is necessary but not sufficient on its own, since the host anchor and attack-channel training between them supply the properties that a key mask on its own cannot deliver.

\paragraph{The host anchor.} Removing the anchor while keeping the key raises the transplant verify rate from \rForge\ to 1.00, so every transplanted carrier verifies once the carrier has collapsed into one reusable pattern (cross-host similarity 1.00 against 0.53, Table~\ref{tab:hostbound}). The anchor is more than a security mechanism, as the host-independent variant also loses most of its payload under the codec (EnCodec BA 0.52 against \rCodec) while presence detection stays at 1.00 in both, so conditioning on the host helps the payload survive re-synthesis rather than hindering it. Its only price is transparency, where a PESQ of \rPESQ\ stands against 3.64 for the ablated model.

\paragraph{Re-synthesis in training.} Robustness to re-synthesis comes from the attack channel, not the architecture. Dropping the codec from that channel lowers EnCodec bit accuracy from \rCodec\ to 0.54 and costs presence detection too (DA 0.92 against 1.00), while the denoiser, which remains in the channel, is unaffected. The ablated model is also quieter at a PESQ of 3.45, since a mark never trained to survive re-synthesis need not carry the energy that survival demands.

\paragraph{Payload length.} Capacity trades off against robustness. \name{} recovers both 8 and 16 bits exactly on clean audio and through the denoiser, whereas at 32 and 64 bits accuracy falls to 0.72 and 0.64 and does not recover with further training, which places reliable capacity at roughly 16 bits.

\begin{table}[t]
\caption{Ablation of \name{}'s mechanisms. Each row removes a single mechanism with all else held fixed.}
\label{tab:ablation}
\centering
\small
\begin{tabular*}{\columnwidth}{@{\extracolsep{\fill}}lccccc c}
\toprule
& Transpl.$\downarrow$ & \multicolumn{2}{c}{EnCodec} & \multicolumn{2}{c}{Denoiser} & PESQ$\uparrow$ \\
\cmidrule(lr){3-4}\cmidrule(lr){5-6}
& (verify) & BA$\uparrow$ & DA$\uparrow$ & BA$\uparrow$ & DA$\uparrow$ & \\
\midrule
\textbf{\name{}} & \textbf{0.04} & \textbf{0.98} & \textbf{1.00} & \textbf{0.98} & \textbf{1.00} & 3.03 \\
\ \ w/o anchor & 1.00 & 0.52 & 1.00 & 0.97 & 1.00 & 3.64 \\
\ \ w/o codec & 0.05 & 0.54 & 0.92 & 0.98 & 1.00 & 3.45 \\
\bottomrule
\end{tabular*}
\end{table}

\section{Discussion}

\paragraph{Imperceptibility and re-synthesis robustness.} Robustness to re-synthesis trades against imperceptibility for a structural reason, since a psychoacoustic loss earns transparency by placing the carrier where the host masks it, which is also where a codec or denoiser discards it. Sweeping \name{}'s SNR budget traces the trade, with denoiser bit accuracy falling from 0.998 to 0.792 as the mark grows quieter while the host-bound carrier holds detection far above the point where the baselines collapse, at a PESQ of \rPESQ.

\paragraph{Limitations.} The principal open limitation is adaptive removal, which falls outside the wrong-key bound. A learned remover trained to reconstruct un-watermarked from marked audio, with access to neither the detector nor the key, drives detection to \rHeldTPR\ within twenty thousand steps and leaves no audible degradation. This does not contradict our re-synthesis results, which concern fixed non-adaptive transformations, and it is the general predicament of all post-hoc watermarking once an attacker can train against the deployed detector. Our evidence is read English speech at one-second scale, so other domains and longer clips still await study.

\paragraph{Deployment scope.} \name{} is not always the right choice, since benign processing with no adversary is served more cheaply by a transparent mark. Its setting is one where a party may read, remove, or fabricate an attribution and keyless designs supply neither no-key protection nor transplant resistance, so an authorized verifier gains a keyed, host-bound rule with a wrong-key acceptance bound (Proposition~\ref{prop:keyed_soundness}) rather than the single passed test that a keyless design is able to offer in its place.

\section{Conclusion}
Distortion-only evaluation does not suffice for speech provenance, where an adversary may read, remove, or transfer a mark. \name{} answers with a key mask, a host-bound carrier, and attack-channel training. On read speech it keeps detection under re-synthesis (spectral denoiser DA~=~1.00, BA~=~\rDenoise), decodes at chance without the key, and rejects transplanted carriers, with detection transferring to held-out DAC and BigVGAN as exact payload recovery degrades. Adaptive removal remains the open limitation, an unbounded attack stripping the mark yet leaving the audio intact. We therefore read keyed, host-bound watermarking as a step toward treating provenance as the adversarial attribution problem it is, rather than as distortion-robust payload recovery.

\bibliographystyle{ACM-Reference-Format}
\bibliography{refs}

\clearpage
\appendix
\section{Appendix}

\subsection{False-Positive Control}
DA calibrates $\tau$ on clean unmarked audio, so we also score unmarked clips under each attack and report the realized false-positive rate at that same threshold (Table~\ref{tab:fpr}). Seven of the fourteen attacks realize a rate at or below 0.013, so their DA means what it says. The other seven sit between 0.030 and 0.127, with EnCodec lowest and the denoiser highest, while 10~dB noise reaches 0.803, a point at which a reported DA of 1.000 is in fact realized at roughly an 80\% false-positive rate rather than 1\%. None of this costs detection, since the matched-threshold DA is 1.000 for every one of these attacks, which shift the absolute score level while preserving the margin between marked and unmarked audio. Under the denoiser the two distributions do not overlap at all, as the lowest-scoring marked clip still outscores the highest-scoring unmarked one, so enforcing a true 1\% rate, or even 0.1\%, leaves DA unchanged. Among the baselines only AudioSeal drifts comparably, and it does so more severely, as its EnCodec DA of 0.85 sits at a 0.58 false-positive rate and falls to 0.09 once matched, leaving most of that apparent detection a codec artifact rather than a recovered watermark. WavMark and Timbre hold the rate near zero on the reported perturbations, because their decode-derived scores cannot false-positive on unmarked audio, which carries no payload for them to recover.

\begin{table}[b]
\caption{False-positive control for \name{}. Values marked $\dagger$ exceed twice the nominal rate.}
\label{tab:fpr}
\centering
\begin{tabular*}{\columnwidth}{@{\extracolsep{\fill}}lccc}
\toprule
Attack & DA (reported) & FPR & DA (matched) \\
\midrule
Clean & 1.000 & 0.007 & 1.000 \\
\multicolumn{4}{@{}l}{\emph{Signal distortion}}\\
Noise (20\,dB) & 1.000 & 0.047$^\dagger$ & 1.000 \\
Noise (10\,dB) & 1.000 & 0.803$^\dagger$ & 1.000 \\
Lowpass (4\,kHz) & 1.000 & 0.007 & 1.000 \\
Highpass (500\,Hz) & 1.000 & 0.103$^\dagger$ & 1.000 \\
Resample (0.8) & 1.000 & 0.013 & 1.000 \\
Gain ($\times0.5$) & 1.000 & 0.000 & 1.000 \\
Quantize (8-bit) & 1.000 & 0.007 & 1.000 \\
MP3 (32\,kbps) & 1.000 & 0.007 & 1.000 \\
Smoothing & 1.000 & 0.003 & 1.000 \\
Echo & 1.000 & 0.050$^\dagger$ & 1.000 \\
\multicolumn{4}{@{}l}{\emph{Re-synthesis}}\\
EnCodec (6\,kbps) & 1.000 & 0.030$^\dagger$ & 1.000 \\
Griffin-Lim & 1.000 & 0.053$^\dagger$ & 1.000 \\
Denoiser & 1.000 & 0.127$^\dagger$ & 1.000 \\
\bottomrule
\end{tabular*}
\end{table}

\subsection{No-Key Access and the Transplant Protocol}
Table~\ref{tab:security} collects the two adversarial axes on which \name{} differs from the baselines. All three baselines are keyless, so their payload is readable by any party who runs the public decoder, and the no-key column records that as \emph{public} rather than as a number. The transplant attack extracts the carrier residual $\delta=x_w-x$ from a marked clip together with its clean host, adds it unscaled to a different host, and tests verification. Because recovering $\delta$ requires the clean host, this is an oracle-residual attack stronger than the stated threat model, which grants the adversary only the marked recording, and \name{}'s \rForge\ transplant verify rate therefore holds even under that attacker-favorable setting. Verification requires both conjuncts of $V_A$, and the presence conjunct uses each method's own 1\%-FPR threshold rather than a shared constant, since presence scores across methods are on different scales. AudioSeal needs a note of its own, since its residual is not a low-energy carrier at all (host-to-residual SNR 0.9~dB against 24.6 and 37.4~dB for \name{} and WavMark), so a transplant carries over part of the source recording along with the mark and its rate should be read as an upper bound.

\begin{table}[t]
\caption{No-key decoding and false attribution.}
\label{tab:security}
\centering
\begin{tabular}{lcc}
\toprule
Method & No-key acc.$\downarrow$ & Transplant verify$\downarrow$ \\
\midrule
WavMark & public & \textbf{0.000} \\
AudioSeal & public & 0.587 \\
Timbre & public & 0.227 \\
\midrule
\textbf{KeyBound (ours)} & \textbf{0.512} & 0.043 \\
\bottomrule
\end{tabular}
\end{table}

\subsection{Perceptual Cost of the Mark}
Table~\ref{tab:quality} reports what each mark costs on clean marked audio, scored against the unmarked host. WavMark is the most transparent scheme on all three measures, and the dagger records that it is evaluated on the 1.1~second clips its decoder requires rather than the 1.0~second clips the other methods use, matching the protocol of Table~\ref{tab:robust}. \name{} sits lowest on PESQ at \rPESQ{} while holding the second-highest STOI at 0.981 and an SNR of 24.6~dB, so its carrier is paid for in perceived quality while intelligibility stays close to the best of the four. Timbre reverses that profile, reaching a PESQ of 3.55 on an STOI of 0.715, the lowest intelligibility here by a wide margin. Read alongside Table~\ref{tab:robust}, this table gives the price side of the robustness reported there, and the operating-point sweep below traces how that price moves as the SNR budget of the mark is tightened or relaxed in either direction.

\begin{table}[t]
\caption{Imperceptibility. $^\dagger$ measured on 1.1\,s clips.}
\label{tab:quality}
\centering
\begin{tabular}{lccc}
\toprule
Method & PESQ$\uparrow$ & STOI$\uparrow$ & SNR (dB)$\uparrow$ \\
\midrule
WavMark$^{\dagger}$ & \textbf{4.18} & \textbf{0.993} & \textbf{37.4} \\
AudioSeal & 3.44 & 0.968 & 18.2 \\
Timbre & 3.55 & 0.715 & 27.7 \\
\midrule
\textbf{KeyBound (ours)} & 3.03 & 0.981 & 24.6 \\
\bottomrule
\end{tabular}
\end{table}

\subsection{Attack Configuration}
Table~\ref{tab:attackcfg} lists the attack parameters read from the evaluation code. All of them operate at 16~kHz except EnCodec, which resamples internally to 24~kHz and back again across the course of its round trip. The two noise levels share one row, as do the two spectral operations built on the same 1024/256 STFT, namely the Griffin-Lim vocoder, which discards phase and re-synthesizes from magnitude alone, and the denoiser, which estimates a noise floor and gates against it. The learned remover in the last row is the separately trained reconstruction network of the limitation study, which sees marked audio and clean targets and never sees either the detector or the key that produced them.

\begin{table}[t]
\caption{Attack configuration, read from the evaluation code.}
\label{tab:attackcfg}
\centering
\begin{tabular*}{0.9\columnwidth}{@{\extracolsep{\fill}}lll}
\toprule
Attack & Parameter / model & Strength \\
\midrule
Noise & additive Gaussian & SNR 20/10 dB \\
Lowpass & biquad low-pass & cutoff 4 kHz \\
Highpass & biquad high-pass & cutoff 500 Hz \\
Resample & rational resample & factor 0.8 \\
Gain & amplitude scale & $\times0.5$ \\
Quantize & bit-depth reduction & 8-bit \\
MP3 & LAME codec & 32 kbps \\
Smoothing & moving average & window 7 \\
Echo & delay $+$ decay & 0.1 s, 0.4 \\
EnCodec & \texttt{encodec\_24khz} & 6 kbps \\
Griffin-Lim & STFT 1024/256 & magnitude only \\
Denoiser & spectral gating & STFT 1024/256 \\
Remover & SEANet U-net (ch 24) & marked$\to$clean \\
\bottomrule
\end{tabular*}
\end{table}

\subsection{Held-Out Re-Synthesis}
EnCodec and the spectral denoiser appear in \name{}'s training attack channel, whereas DAC \citep[Descript Audio Codec,][official 16\,kHz checkpoint, all 12 quantizers]{dac} and BigVGAN \citep[][\texttt{nvidia/\allowbreak bigvgan\_\allowbreak v2\_\allowbreak 24khz\_\allowbreak 100band\_\allowbreak 256x}]{bigvgan} do not and therefore test attack generalization (Table~\ref{tab:heldout}). Presence detection transfers well to both, reaching DA 0.92 and 1.00 at the fixed LibriSpeech presence threshold, which we do not recalibrate and which agrees with the ROC TPR@1\%FPR to rounding, while exact payload recovery degrades to BA 0.70 and 0.50, which reproduces on unseen channels the same two-level pattern that Griffin-Lim already shows in Table~\ref{tab:robust}. The three quality columns compare the attacked output to the clean host, so they measure how much each channel damages the audio rather than how imperceptible the watermark is. SI-SNR is time-aligned by an integer cross-correlation lag, which matters for DAC, whose fixed 8-sample codec delay moves it from -15.0 raw to 10.5 aligned, whereas BigVGAN stays negative after alignment (-2.9 $\to$ -2.7) because mel re-synthesis rewrites the phase it cannot recover. Raw unaligned values are in the audit file.

\begin{table}[t]
\caption{Held-out re-synthesis on the same 300 \texttt{test-\allowbreak clean} clips, neither channel seen in training.}
\label{tab:heldout}
\centering
\footnotesize
\begin{tabular*}{0.8\columnwidth}{@{\extracolsep{\fill}}lccccc}
\toprule
Attack & DA & BA & PESQ & STOI & SI-SNR \\
\midrule
DAC & 0.92 & 0.70 & 2.99 & 0.96 & 10.5 \\
BigVGAN & 1.00 & 0.50 & 2.88 & 0.98 & -2.7 \\
\bottomrule
\end{tabular*}
\end{table}

\subsection{Operating Point and Adaptive Removal}
Figure~\ref{fig:oppoint_app}(a) traces \name{}'s detection--quality trade-off as the SNR budget varies under a fixed 30,000-step schedule, where a quieter mark reaches higher PESQ while denoiser bit accuracy falls from 0.998 at the loudest setting to 0.792 at the quietest, and the deployed model sits between them at a PESQ of \rPESQ\ with \rDenoise\ denoiser bit accuracy. Figure~\ref{fig:oppoint_app}(b) turns to the adaptive remover, against which detection survives a weak attacker (1.00 at one thousand steps) before falling to 0.17 by four thousand and to \rHeldTPR\ by twenty thousand. Table~\ref{tab:remover} adds the quality of the removed audio over the same 200 clips, which is what makes this a genuine limitation rather than a nuisance, since once detection has collapsed at four thousand steps the remover output is closer to the clean host than the marked audio ever was (PESQ 4.18 to 4.42 against \rPESQ\ for the mark itself), so removal costs the attacker nothing in audible quality and simply leaves the mark gone.

\begin{figure*}[t]
\centering
\includegraphics[width=0.62\textwidth]{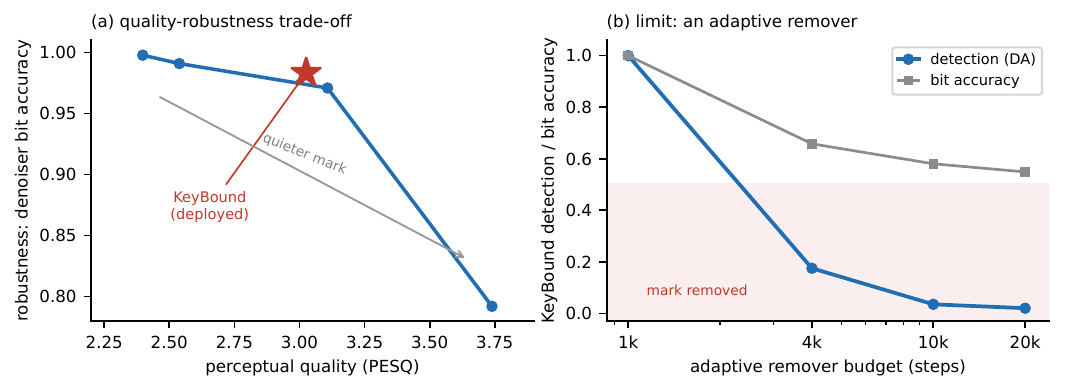}
\caption{(a)~Sweeping the SNR budget traces the quality--robustness frontier, and the deployed model (star) sits at a PESQ of \rPESQ\ with \rDenoise\ denoiser bit accuracy. (b)~Detection and bit accuracy under a learned reconstruction remover (marked$\to$clean, no detector access) of increasing training budget. The shaded region ($\mathrm{DA}<0.5$) marks a removed watermark.}
\label{fig:oppoint_app}
\end{figure*}

\begin{table}[t]
\caption{Learned reconstruction remover across training budgets, with quality measured against the clean host.}
\label{tab:remover}
\centering
\small
\begin{tabular*}{0.85\columnwidth}{@{\extracolsep{\fill}}lccccc}
\toprule
Budget & DA & BA & PESQ$\uparrow$ & STOI$\uparrow$ & SNR$\uparrow$ \\
\midrule
1k & 1.00 & 1.00 & 3.02 & 0.98 & 24.6 \\
4k & 0.17 & 0.66 & 4.18 & 0.99 & 32.8 \\
10k & 0.04 & 0.58 & 4.37 & 0.99 & 35.4 \\
20k & 0.02 & 0.55 & 4.42 & 0.99 & 36.9 \\
\bottomrule
\end{tabular*}
\end{table}

\subsection{Compound Attacks}
Chaining two transformations is a stronger threat than either applied alone. Table~\ref{tab:compound} reports detection accuracy at the clean-calibrated threshold under four such chains, on all of which \name{} keeps DA~=~1.00, while AudioSeal, already fragile under a single re-synthesis, drops to 0.41 and below. Its worst chain ends in vocoder re-synthesis and reaches 0.05, close to the 0.08 it scores under Griffin-Lim alone in Table~\ref{tab:robust}, so the chain that finishes on its weakest single channel is the one that costs it the most.

\begin{table}[t]
\caption{Detection accuracy under two-stage attacks.}
\label{tab:compound}
\centering
\begin{tabular}{lcc}
\toprule
Compound attack & \name{} & AudioSeal \\
\midrule
EnCodec $\to$ noise & 1.00 & 0.41 \\
Denoiser $\to$ resample & 1.00 & 0.25 \\
EnCodec $\to$ denoiser & 1.00 & 0.22 \\
MP3 $\to$ Griffin-Lim & 1.00 & 0.05 \\
\bottomrule
\end{tabular}
\end{table}

\subsection{Host Binding and Wrong-Key Decoding}
Table~\ref{tab:hostbound} checks host dependence directly. Holding payload and key fixed, \name{}'s carriers across 300 host clips show a mean pairwise correlation of 0.53, well below the w/o-anchor variant, whose carriers are identical at correlation 1.00 and thus reduce to one reusable pattern. Cosine similarity gives the same ordering at the same two values, so the gap does not depend on which similarity we read. This is consistent with the transplant gap in Table~\ref{tab:security}, although transplant verification remains the direct false-attribution metric and these similarities support host dependence empirically rather than proving it. Table~\ref{tab:nokey} reports keyed access over 300 marked clips, where the payload decodes exactly under the correct key and at chance under a zero mask or a random wrong key, the latter averaged over $K{=}10$ wrong keys per clip with the spread given as a standard deviation. A bit accuracy near 0.5 means the authorized decoder recovers no payload without the correct mask, which is the wrong-key behaviour that our no-key challenge measures, and it carries no claim about the multi-message secrecy of the marked waveform itself.

\begin{table}[t]
\caption{Cross-host carrier similarity at a fixed payload and key, over 300 host clips.}
\label{tab:hostbound}
\centering
\small
\begin{tabular}{lcc}
\toprule
Method & Carrier corr.$\downarrow$ & Cosine sim.$\downarrow$ \\
\midrule
\name{} & 0.53 & 0.53 \\
w/o anchor & 1.00 & 1.00 \\
\bottomrule
\end{tabular}
\end{table}

\begin{table}[t]
\caption{No-key and wrong-key decoding of \name{}.}
\label{tab:nokey}
\centering
\small
\begin{tabular}{lcc}
\toprule
Decoding setting & Bit accuracy & Interpretation \\
\midrule
Correct key & 1.00 & exact payload recovery \\
No key / zero mask & 0.51 & chance level \\
Random wrong key & $0.50\pm0.04$ & chance level \\
\bottomrule
\end{tabular}
\end{table}

\subsection{Wrong-Key Acceptance and Payload Length}
Proposition~\ref{prop:keyed_soundness} bounds the probability that the attribution rule $V_A$ accepts a decode made with the wrong key, and we restate here the constant and its scaling. Without $\kappa_A$ the keyed decode $\hat m(y,\kappa_A)$ is uniform on $\{0,1\}^n$ in the adversary's view (Proposition~\ref{prop:single_sample}, Corollary~\ref{cor:chosen_payload_attribution}), so a false accept demands that a uniform draw fall within Hamming distance $\rho$ of the registered payload $m_A^\star$,
\begin{equation*}
  \beta(n,\rho)=\Pr_{u\sim\mathcal U\{0,1\}^n} \big[d_H(u,m_A^\star)\le\rho\big]=2^{-n}\sum_{i=0}^{\rho}\binom{n}{i}.
\end{equation*}
The presence gate $\mathbf 1[s(y)\ge\tau_s]$ multiplies this by a value in $\{0,1\}$ and can only lower it, so the bound holds whatever the presence threshold and whatever attack inflates the presence-only false-positive rate. At the deployed operating point ($n=16$, $\rho=2$) this gives $\beta=137/2^{16}\approx2.1\times10^{-3}$, and because $\beta$ falls combinatorially in $n$ at fixed $\rho$ (Table~\ref{tab:beta}), the bound tightens by almost two orders of magnitude for every eight bits added. This is the statistical face of the capacity trade-off seen in the payload-length sweep, where an 8-bit payload is too short to bound anything meaningfully while lengths beyond 16 demand a payload robustness the carrier cannot sustain. The tolerance $\rho$ moves the trade the other way, since a larger $\rho$ raises the true-accept rate of $V_A$ and enlarges $\beta$ along with it, so we deploy $\rho=2$ throughout.

\begin{table}[t]
\caption{Wrong-key acceptance bound $\beta(n,\rho{=}2)$ by payload length. Lower is stronger.}
\label{tab:beta}
\centering
\small
\begin{tabular}{lcccc}
\toprule
$n$ (bits) & 8 & 16 & 24 & 32 \\
\midrule
$\beta$ & $1.4{\times}10^{-1}$ & $2.1{\times}10^{-3}$ & $1.8{\times}10^{-5}$ & $1.2{\times}10^{-7}$ \\
\bottomrule
\end{tabular}
\end{table}

\begin{table}[t]
\caption{Confidence intervals (95\%) over the test set.}
\label{tab:ci}
\centering
\footnotesize
\begin{tabular}{llc}
\toprule
Metric & Value & 95\% CI \\
\midrule
\multicolumn{3}{@{}l}{\emph{Detection / security rates (Clopper--Pearson)}}\\
\name{} DA, denoiser & 1.000 & [0.988, 1.000] \\
\name{} DA, EnCodec & 1.000 & [0.988, 1.000] \\
AudioSeal DA, EnCodec & 0.850 & [0.804, 0.888] \\
\name{} transplant verify & 0.043 & [0.023, 0.073] \\
AudioSeal transplant verify & 0.587 & [0.529, 0.643] \\
\midrule
\multicolumn{3}{@{}l}{\emph{Mean bit accuracy (clip-level bootstrap)}}\\
\name{} BA, denoiser & 0.985 & [0.980, 0.989] \\
\name{} BA, EnCodec & 0.983 & [0.979, 0.986] \\
\name{} no-key BA & 0.499 & [0.485, 0.513] \\
AudioSeal BA, denoiser & 0.553 & [0.538, 0.568] \\
\bottomrule
\end{tabular}
\end{table}

\subsection{Confidence Intervals}
Table~\ref{tab:ci} reports 95\% confidence intervals over the test clips, with exact Clopper--Pearson intervals for the binary rates and a 2000-resample clip-level bootstrap for the mean bit accuracies. They support the trends the aggregate tables report, with detection at 1.00 under codec and denoiser and no-key decoding at chance, and no \name{} interval meets the matching AudioSeal one on detection, on denoiser bit accuracy, or on transplant verification. A dedicated per-clip pass produces them, so a point estimate may differ from the aggregate tables in the third decimal wherever the per-clip key draw affects the quantity being measured.

\end{document}